%% file: main.tex
\documentclass[aps,prl,reprint,amsmath,amssymb,longbibliography,nofootinbib,superscriptaddress]{revtex4-2}
\usepackage[T1]{fontenc}
\usepackage{lmodern}
\usepackage{graphicx}
\usepackage{bm}
\usepackage{mathtools}
\usepackage{microtype}
\usepackage[colorlinks=true,citecolor=blue!55!black,linkcolor=blue!55!black,urlcolor=blue!55!black]{hyperref}
\usepackage{xcolor}
\usepackage{amsmath}
\usepackage{tikz}

\newtheorem{theorem}{Theorem}
\newcommand{\cH}{\mathcal H}
\newcommand{\cD}{\mathcal D}
\newcommand{\Tr}{\operatorname{Tr}}
\newcommand{\Ran}{\operatorname{Ran}}
\newcommand{\vect}{\operatorname{vec}}

\begin{document}

\title{Topology Obstructs Pure Foundation Neural Quantum States}

\author{Timothy Heightman}
\email{theightman@icfo.net}
\affiliation{ICFO--Institut de Ci\`encies Fot\`oniques, The Barcelona Institute of Science and Technology, 08860 Castelldefels, Barcelona, Spain}
\affiliation{Simulacra Research Inc., London, UK and Chicago, USA}
\author{Elena Orlova}
\affiliation{Simulacra Research Inc., London, UK and Chicago, USA}
\author{Philip Mantrov}
\affiliation{Simulacra Research Inc., London, UK and Chicago, USA}
\author{Aleksei Ustimenko}
\email{aleksei@simulacra-ai.com}
\affiliation{Simulacra Research Inc., London, UK and Chicago, USA}


\begin{abstract}
Foundation models for ground states in spin-1/2 systems are a promising method for problems ranging from quantum chemistry to identifying new phase diagrams. Nearly all such models are currently pure-states that condition on the Hamiltonian's parameters, whose Monte Carlo samples give energy  estimates according to the variational principle. In this contribution, we show that this representation is topologically obstructed. For any gapped Hamiltonian family whose ground-state bundle is non-trivial, every continuous normalized state-vector model has zero fidelity with the ground state at some parameter value in the Hamiltonian family. For that value, the energy is at least one spectral gap, $\Delta$, with an $\mathcal{O}(\Delta)$ gap in an open-neighbourhood of that point. We show that this is a sufficient no-go also in the case of degenerate ground-state manifolds, time dynamics, and periodic systems with mixed space-time topology, demonstrating these obstructions on one- and two-qubit systems. We discuss how this causes a spike in the fidelity susceptibility, giving a numerical signature of a phase-transition where there is none. Finally, we show that operator-valued models canonically avoid these obstructions and preserve topological information, implying a structural necessity in representation for foundation neural quantum states.

\end{abstract}

\maketitle

\begin{figure*}[t]
\centering
\includegraphics[width=0.90\textwidth]{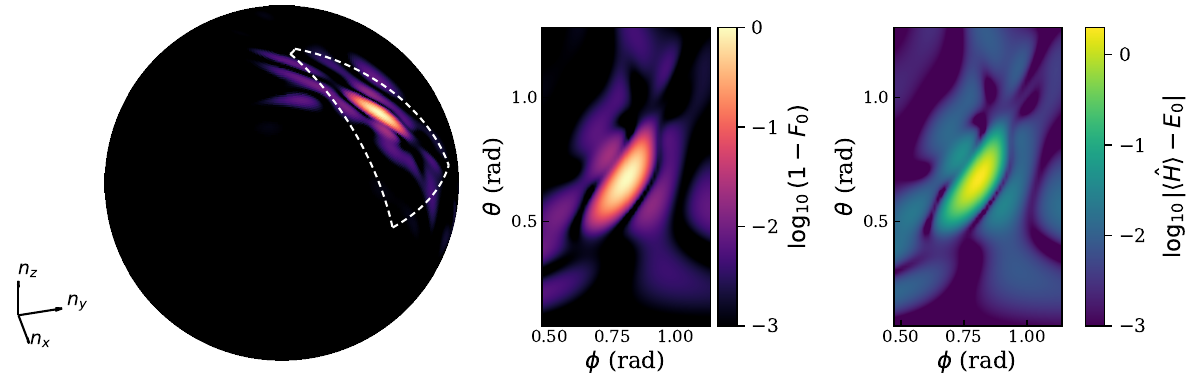}
\caption{\label{fig:static}Fidelity and energy gap to the true ground state of trained wavefunction-valued FNQS. 
The sphere is the Hamiltonian family $\hat H(\bm n)=-\bm n\cdot\hat{\bm\sigma}$, and its heatmap shows $\log_{10}[1-F_0(\bm n)]$ 
after variational training. The white dashed boundary marks the $(\theta,\phi)$ window enlarged in the two right panels, which show $\log_{10}(1-F_0)$ and $\log_{10}|\langle\hat H\rangle-E_0|$, respectively.  
At the maximum $x_*$, $F_0=0$ and the energy error is $2$, one full spectral gap, with an $\mathcal{O}(\Delta)$ gap persisting in the neighbourhood 
of $x_*$, showing this obstruction has a non-zero measure in the total Hamiltonian space.  Architecture and training details are given in the End Matter.}
\end{figure*}

\paragraph{Introduction.---}
Finding the ground state of a quantum many-body Hamiltonian is a central computational problem for quantum technology. It determines electronic structure in quantum chemistry, encodes combinatorial optima, and, across a parameterised family of Hamiltonians, maps the phase diagram in which new many-body physics appears~\cite{bauer2020quantum,mcardle2020quantum,du2013handbook,gu2010fidelity}.  
Exponentially large Hilbert spaces have motivated several complementary strategies for finding ground states, from variational quantum algorithms to tensor networks and neural quantum states (NQS)~\cite{carleo2017solving,lange2024architectures}. The latter has become a leading classical approach, since it is not bound by the area-law constraints of tensor networks~\cite{hastings2007arealaw,cirac2021peps_review}, nor does it natively suffer from the barren plateau phenomenon~\cite{cerezo2021variational,mcclean2018barren}.

Conventionally, the training cost for NQS requires optimizing an architecture from scratch for each given instance of a many-body system, including in electronic NQS such as FermiNet, PauliNet, and DeepErwin~\cite{pfau2020ferminet,hermann2020pauli,gerard2022goldstandard}. Foundation neural quantum states (FNQS)~\cite{zhang2023transformer,rende2024finetuning,viteritti2025spinglass,zaklama2025attention} changed this paradigm, amortizing that cost with one model trained across a Hamiltonian family such as a family of Ising models with differing interaction and transverse field strengths \cite{Rende2025}. So far these models have predominantly been wavefunction-valued. That is, for basis bit strings $b$, a neural network returns $c_\theta(\hat H,b)\in\mathbb C$ and hence $|\psi_\theta(\hat H)\rangle=\sum_b c_\theta(\hat H,b)|b\rangle$, usually accessed by variational Monte-Carlo sampling. Because the network is a continuous function of the Hamiltonian parameters, such a model is a continuous map from the parameter space to non-zero state vectors. This is irrespective of whether a model outputs amplitudes directly, or log-amplitudes and phases. The central aim of FNQS is to maintain a positive overlap with the true ground eigenstate over the entire family, with a performant model having as large an overlap with the true ground state for every Hamiltonian in the family.

In this contribution, we show that this goal is unattainable for topologically non-trivial gapped Hamiltonian families. Regardless of model capacity, architecture, loss function, or optimization procedure, every pure-state FNQS produces becomes exactly orthogonal to the true ground state at some point in the Hamiltonian family's parametric space, giving zero ground-state fidelity and an energy error of at least one spectral gap. Because the fidelity is continuous, this obstruction occupies a finite open region of the parameter space, and its location is fixed by training dynamics. The obstruction arises whenever the family encloses a degeneracy carrying non-zero Berry phase~\cite{Berry1984}, which is the case for many systems of interest such as rotating fields, boundary twists, momentum or time cycles~\cite{Thouless1983,NiuThoulessWu1985,Rudner2013}. At the zero-fidelity node this obstruction produces a spike in the fidelity susceptibility, giving a spurious signature of a phase transition where there is none. We establish this result for both unique and degenerate ground states, and show that it persists under time-dynamics, including driven and periodic systems with mixed space-time topology. Finally, we prove that operator-valued foundation models evade the obstruction and preserve the family's physical topological information, including in dynamical settings. This implies a general representation theoretic design-rule for FNQS when searching for phases of matter that avoids these false-positives.

Parameter-dependent Hamiltonian eigenspaces have long been understood geometrically. Berry and Simon formulated adiabatic eigenstates as vector bundles with connection and holonomy~\cite{Simon1983,Berry1984}. Related obstructions govern smooth periodic Bloch frames in Wannier theory~\cite{Panati2007,Brouder2007}, while Thouless identified mixed momentum--time topology in quantized pumping~\cite{Thouless1983}. These results concern the non-existence of a global eigenvector gauge or frame.

Here, an FNQS poses a weaker approximation problem that is obstructed by the same topology. Its output need not be an exact ground eigenvector, only a normalized state with non-zero ground-state overlap. The rank-one form of this observation first appeared in the analysis of a recent foundation model~\cite{HamiltonZero2026}, and here we treat it in full.

\begin{figure*}[t]
\centering
\includegraphics[width=\textwidth]{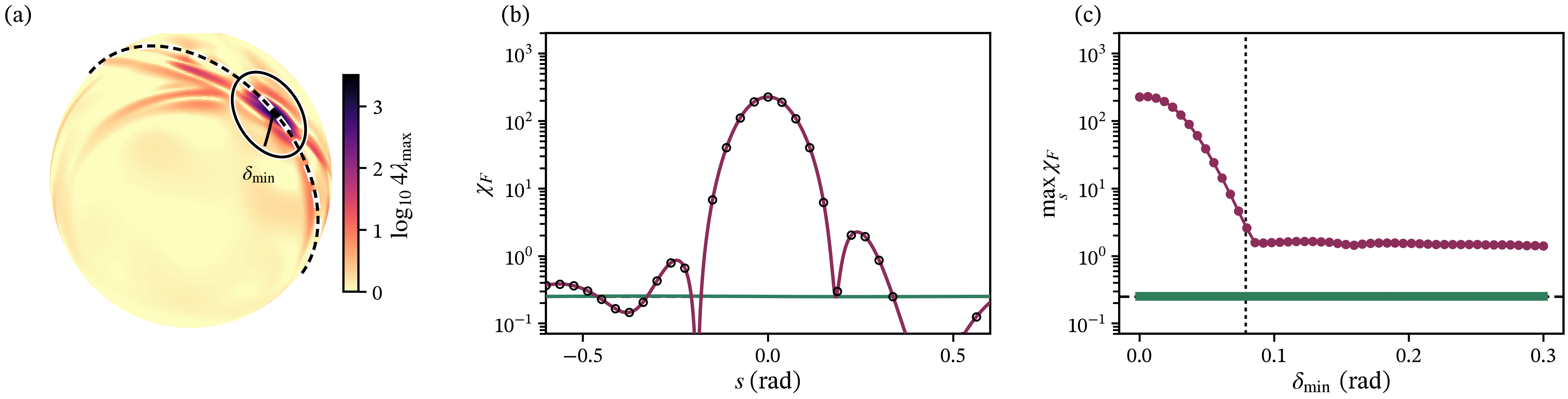}
\caption{\label{fig:supp-chi} Trained-state fidelity susceptibility manufactures a phase-transition signature at the zero-fidelity node of a pure state FNQS. (a) The largest eigenvalue $\lambda_{\text{max}}$ of the model's fidelity metric on the parameter space $X = S^2$. The dashed line shows a great circle that sweeps in the parameter space in (b)-(c). The black ring in (a) shows a geodesic radius $\delta_{\text{min}}=0.3$ about $x_*$, which is the largest closest approach swept in (c).  (b) The Fidelity susceptibility of the model (red) along the great circle in (a) and the analytical value (green), with training points shown as circles. The model's susceptibility spike peaks at $903$ times the exact value of $\chi_{F}=1/4$ along the great circle. (c) Maximum $\chi_F$ values over great circles with closest approach $\delta_{\min}$ to $x_*$, with the dashed line marking the closest approach such that the pure-state model has $50\%$ fidelity to the true ground state. 
The green points show a density-valued model on the same set of paths in $S^2$, with the same number of variational parameters, sitting at $0.2532 \pm 0.0003$, a $1.3\%$ offset of the true value $\chi_F = 1/4$.}
\end{figure*}

\paragraph{Non-degenerate Obstruction.---}

Let $x \in X$ parametrise a continuous family of finite-dimensional Hamiltonians $\hat H(x)$ acting on a Hilbert space $\cH$. We assume that the ground-state energy $E_0(x)$ has fixed degeneracy $r$ and remains separated from the excited spectrum by a gap $\Delta(x)\geq\Delta_{\min}>0$ throughout the family. Because this gap never closes, the corresponding ground-space projector $\hat P_0(x)$ varies continuously with the Hamiltonian~\cite{Kato1995} (see End Matter). The collection of ground states defines the ground-state bundle
\begin{equation}
 E=\Ran \hat P_0\longrightarrow X.
 \label{eq:ground-bundle}
\end{equation}
For a unique ground state, $r=1$, this is the ground-state line bundle $L=E$. Writing any physical state as a density operator $\hat\rho$, the gap gives $\hat H\geq E_0 \hat I+\Delta(\hat I-\hat P_0)$. Tracing against any density operator $\hat\rho$ gives
\begin{equation}
 \Tr(\hat\rho \hat H)-E_0\geq\Delta\Tr[\hat\rho(\hat I-\hat P_0)],
 \label{eq:gapineq}
\end{equation}
hence a state with no ground-space overlap has an energy error of at least one spectral gap.

A normalized pure-state model assigns a state vector $|\Phi(x)\rangle\in\cH$, continuously across $X$, with fidelity to the ground state over the parametric space, $F_0(x)=\langle\Phi(x)|\hat P_0(x)|\Phi(x)\rangle$. If $F_0(x)$ is positive everywhere, then the projected state $\hat P_0(x)|\Phi(x)\rangle$ must be nowhere vanishing over $X$. The following theorem shows that a topologically non-trivial ground-state line bundle makes this requirement impossible to satisfy.

\begin{theorem}
\label{thm:line}
For $r=1$, if the ground-state line bundle $L$ is topologically nontrivial, for example with a nonzero first Chern number, every continuous normalized vector model has a parameter value $x_*$ such that,
\begin{equation}
 \begin{aligned}
 \hat P_0(x_*)|\Phi(x_*)\rangle&=0,\\
 \langle\Phi(x_*)|\hat H(x_*)|\Phi(x_*)\rangle
   &\geq E_0(x_*)+\Delta(x_*).
 \end{aligned}
 \label{eq:linefail}
\end{equation}
Consequently $F_0(x_*)=0$, and the energy error there is at least $\Delta(x_*)\geq\Delta_{\min}$.
\end{theorem}
\emph{Proof.}  If $F_0(x)>0$ everywhere, then
$|\widetilde\psi_0(x)\rangle=\hat P_0(x)|\Phi(x)\rangle/\sqrt{F_0(x)}$
would define a continuous normalized ground state over all of $X$, contradicting the non-triviality of $L$.  Hence $F_0(x_*)=0$ for some $x_* \in X$, and Eq.~\eqref{eq:gapineq} gives Eq.~\eqref{eq:linefail}. $\square$

Hence every $x \in X$ separately has a normalized ground state $|\Phi(x)\rangle$, yet no single continuous vector output can maintain even an arbitrarily small positive ground-state fidelity everywhere. At a zero-fidelity node, the model state lies entirely in the excited subspace and hence its energy is at least one spectral gap above the ground state. Furthermore, since the Fidelity function is continuous, there will always be an open neighbourhood around $x_*$ for which the energy is $\mathcal{O}(\Delta)$. 

We can see this obstruction clearly on a single qubit system. Consider the parameter space $X = S^2$ with $\bm n\in S^2$, and the Hamiltonian
with ground-state projector
\begin{equation}
 \hat H(\bm n)=-\bm n\!\cdot\!\hat{\bm\sigma},
 \qquad \hat P_0=(\hat I+\bm n\!\cdot\!\hat{\bm\sigma})/2.
 \label{eq:hopf}
\end{equation}
Normalized one-qubit vectors form $S^3$, meaning a wavefunction-valued FNQS for this family is a continuous map $\Phi:S^2\to S^3$. For the family in Eq.~\eqref{eq:hopf}, let $L_{\bm n}=\Ran\hat P_0(\bm n)$ be the one-dimensional ground space at $\bm n$. The collection $L=\bigsqcup_{\bm n\in S^2}L_{\bm n}\to S^2$ is the Hopf line bundle, whose first Chern number satisfies $|c_1(L)|=1$.  
Because of this non-trivial topology, Theorem~\ref{thm:line} implies that it must return the excited state at some $\bm n_*$, changing the energy from $-1$ to $+1$. Figure~\ref{fig:static} shows this obstruction in a wavefunction-valued FNQS trained directly on the Hamiltonian family in Eq.~\eqref{eq:hopf}.

\begin{figure*}[t]
\centering
\includegraphics[width=0.78\textwidth]{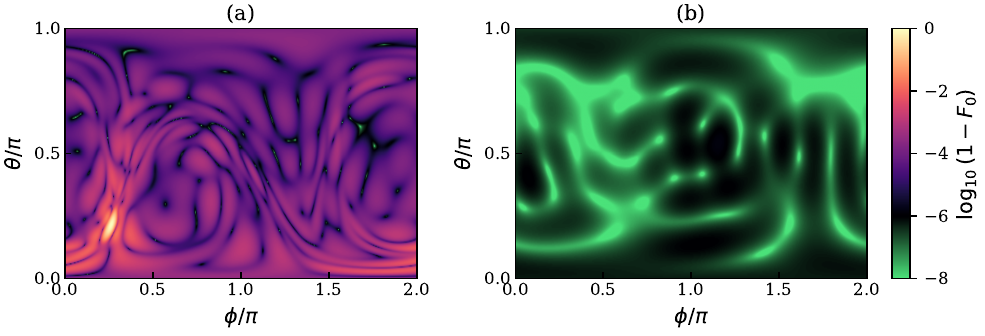}
\caption{\label{fig:density-lift} Matched-capacity density-matrix realization of the operator lift on the one-qubit family of Eq.~\eqref{eq:hopf}. The shared color scale shows $\log_{10}(1-F_0)$ on the same independent $(\theta,\phi)$ grid; a green tail resolves values between $10^{-6}$ and $10^{-8}$. (a) The wavefunction-valued FNQS has a grid maximum $1-F_0=0.996$ and its continuously refined maximum is 1.000. (b) The density-matrix-valued FNQS removes the forced failure: the grid maximum is $1.529\times10^{-6}$ and continuous refinement gives $1.531\times10^{-6}$. The training protocol is matched: both networks use the same Hamiltonian inputs, $3$--$48$--$48$--$3$ architecture, initialization, grid, optimizer, and schedule. Their output parameterisations and independently optimized weights differ.}
\end{figure*}

\paragraph{Degenerate Obstruction.---} Degeneracy might appear to circumvent this obstruction. For a unique ground state, positive fidelity means a model must follow a single ground-state direction continuously across a Hamiltonian family. Within an $r$-fold degenerate ground space, however, the output can rotate among different ground states as the Hamiltonian varies. This freedom might allow a continuous ground-state choice even when the full ground-space bundle is nontrivial. The following theorem however shows that topology obstructs the degenerate case too.

\begin{theorem}
\label{thm:rankr}
A continuous normalized vector model can satisfy $F_0(x)>0$ for every $x\in X$ if and only if the ground-state bundle $E$ contains a trivial line sub-bundle. Consequently, if $E$ contains no such line, every continuous normalized vector model has $F_0(x_*)=0$ for some $x_* \in X$, with an energy error of at least one spectral gap. The condition $c_r(E)\neq0$ is a sufficient topological witness.
\end{theorem}
\emph{Proof.} Assume for contradiction that $E$ contains no trivial line subbundle, but that a continuous normalized model satisfies $F_0(x)>0$ throughout $X$. Then
\[
|\widetilde\psi_0(x)\rangle
=\frac{\hat P_0(x)|\Phi(x)\rangle}{\sqrt{F_0(x)}}
\]
is a continuous normalized ground state. Its span defines a line subbundle $L\subset E$. Because $|\widetilde\psi_0(x)\rangle$ is a global nonvanishing section of $L$, that line bundle is trivial, which is a contradiction.
Conversely, if $E$ contains a trivial line subbundle, that line has a continuous normalized ground-state section. Choosing this section as the model output gives $F_0(x)=1$ throughout $X$. 
Finally, suppose $c_r(E)\neq0$. If $E$ contained a trivial line $L$, its orthogonal complement $E'$ would give
\[
 E\simeq L\oplus E',\qquad \operatorname{rank}E'=r-1.
\]
Since $L$ is trivial,  $c(L)=1$. the Whitney formula~\cite[Sec.~14.4]{MilnorStasheff1974} gives $c(E)=c(E')$, and therefore $c_r(E)=0$, because a rank-$(r-1)$ bundle has no $r$th Chern class. This contradiction shows that $c_r(E)\neq0$ rules out an everywhere-positive fidelity. Equation~(2) then gives the spectral-gap energy bound at the zero-fidelity point. \hfill$\square$

Restricting $X$ to a gapped patch with a trivial ground-state bundle removes this no-go, but also narrows the Hamiltonian family represented by the model. If the model domain contains even one closed gapped subfamily $Y\subset X$ with $c_r(E|_Y)\neq0$, its restriction to $Y$ already has a zero-fidelity node. 

As an example of the degenerate obstruction, consider two qubits with $X=S^2\times S^2$, with $\hat Q_\pm=(\hat{I}^{\otimes 2}\pm \hat{I} \otimes \hat \sigma^z_2)/2$
and the Hamiltonian family,
\begin{equation}
 \hat H(\bm n_+,\bm n_-)=-\hat Q_+\,\bm n_+\!\cdot\!\hat{\bm\sigma}_1
                         -\hat Q_-\,\bm n_-\!\cdot\!\hat{\bm\sigma}_1.
 \label{eq:twoqubit}
\end{equation}
The two $\hat{I} \otimes\hat \sigma_2$ sectors each contribute one energy-$-1$ ground state, so $r=2$ and the gap is $2$ in natural units. Each sector carries a Hopf line bundle $L_\pm$ with first Chern number $|C_1(L_\pm)|=1$ on its sphere. Hence the full rank-two ground space is the direct sum~\cite{MilnorStasheff1974}, $E=L_+\oplus L_-$. For a direct sum, the second Chern number is the product of the two first Chern numbers \cite{MilnorStasheff1974},
\begin{equation}
 C_2(E)\equiv\int_{S^2\times S^2}c_2(E) =C_1(L_+)C_1(L_-),
 \label{eq:twoqubit-chern}
\end{equation}
and hence $|C_2(E)|=1\neq0$. Any continuous normalized wavefunction-valued FNQS defined over the full $S^2\times S^2$ Hamiltonian family is therefore orthogonal to the two-dimensional ground space at some $(\bm n_+,\bm n_-)$, where its energy is $+1$.

\emph{Fidelity susceptibility of an obstructed model.---}
A standard model-based numerical signature of criticality is the fidelity susceptibility of the learned family along a parameter path $\lambda\mapsto x(\lambda) \subseteq X$~\cite{gu2010fidelity,rende2025foundation,sinibaldi2026designing},
\begin{equation}
 \chi_F(\lambda)=\lim_{\delta\lambda\to0}\frac{2\left[1-\left|\langle\psi_\theta(\lambda)|\psi_\theta(\lambda+\delta\lambda)\rangle\right|\right]}{\delta\lambda^2}
 \label{eq:supp-chi}
\end{equation}
Let us consider again the one-qubit example above, in which Eq.~\ref{eq:supp-chi} reduces to $\frac14\,|\partial_\lambda \bm b_\theta|^2$ for a Bloch vector $\bm b_\theta$. For this system's Hamiltonian (Eq.~\ref{eq:hopf}), there is no phase transition on the parameter space $S^2$. This is because the Fubini--Study (FS) metric is $\tfrac14(d\theta^2+\sin^2\theta\,d\phi^2)$. Hence along any unit-speed great circle $\chi_F=1/4$ exactly, and the gap is $2$ everywhere. A continuous wavefunction-valued model must nonetheless vanish against the target at some $x_*$ by Theorem~\ref{thm:line}, which induces a spurious spike in the Fidelity susceptibility, as we show below.

Let $\lambda\in[0,\epsilon]$ parametrise a path beginning at the node $x_*$, and suppose the ground-state fidelity recovers from $F_0(0)=0$ to $F_0(\epsilon)=1/2$. The model--target Fubini-Study distance therefore falls from $\pi/2$ to $\pi/4$. If $\chi_\theta$ and $\chi_0$ are the model and exact ground-state susceptibilities along the path, the triangle inequality gives
\begin{equation}
    \frac{\pi}{4}
    \leq
    \int_0^\epsilon
    \left(\sqrt{\chi_\theta(\lambda)}
    +\sqrt{\chi_0(\lambda)}\right)\,\mathrm{d}\lambda.
\end{equation}
Defining $M_\mu=\max_{\lambda\in[0,\epsilon]}\chi_\mu(\lambda)$ for $\mu\in\{\theta,0\}$, we obtain directly
\begin{equation}
    \sqrt{M_\theta}+\sqrt{M_0}
    \geq \frac{\pi}{4\epsilon}.
\end{equation}
For the one-qubit family in Eq.~(4), $M_0=1/4$ along a unit-speed great circle, and hence
\begin{equation}
    \sqrt{M_\theta}
    \geq \frac{\pi}{4\epsilon}-\frac{1}{2}.
\end{equation}
For $\epsilon<\pi/4$, this bound exceeds the exact susceptibility $M_0=1/4$, inducing a spurious peak, despite the exact family having no phase transition, see in Fig.~\ref{fig:supp-chi}.

\paragraph{Operator-valued models are unobstructed.---} When using operator-valued models that assign an operator $\hat A(x)\in\operatorname{End}(\cH)$ for each $x\in X$, the obstruction is lifted. This is because an operator-valued model is in the same class as the ground-state projection operator that varies smoothly in the Hamiltonian's parametric space. To that end, let $\hat\rho_A=\hat A\hat A^\dagger$ defines the physical state, with $\Tr\hat\rho_A=1$, and the ground-state fidelity $F_0^{\rm op}=\Tr(\hat\rho_A\hat P_0)=\Tr(\hat A^\dagger\hat P_0\hat A)$. Eq.~\eqref{eq:gapineq} applies to $\hat\rho_A$ unchanged, so $\Tr(\hat\rho_A\hat H)-E_0\geq\Delta(1-F_0^{\rm op})$. Recall also for any operator $\hat A$ on $\cH$, we may define its vectorization as the vector
\begin{equation}
 |\hat A\rangle\!\rangle\equiv\vect{\hat A}
 =\sum_j|\bar e_j\rangle\otimes \hat A|e_j\rangle
 \in\bar\cH\otimes\cH,
 \label{eq:vecdef}
\end{equation}
where $\{|e_j\rangle\}$ is any orthonormal basis.

\begin{theorem}
\label{thm:lift}
For every continuous gapped Hamiltonian family with fixed ground-state degeneracy $r$, there exists a continuous operator-valued map
\(
\hat A_0:X\longrightarrow\operatorname{End}(\cH)
\)
with unit Hilbert--Schmidt norm such that the fidelity
\(
F_0^{\text{op}}(x)=1,
\)
for all $x\in X$, and
\(
\hat H(x)\hat A_0(x)=E_0(x)\hat A_0(x)
\)
for every $x\in X$.
Its vectorization $|T_0(x)\rangle=|\hat A_0(x)\rangle\!\rangle$ is a continuous normalized vector in $\overline{\cH}\otimes\cH$ satisfying
\begin{equation}
    (\hat I\otimes\hat H(x))|T_0(x)\rangle
    =E_0(x)|T_0(x)\rangle.
\label{eq:lift}
\end{equation}

\end{theorem}
\emph{Proof.} Choosing
\[
\hat A_0(x)=\frac{\hat P_0(x)}{\sqrt r},
\]
then continuity of $\hat P_0(x)$ makes $\hat A_0(x)$ continuous. Since $\hat P_0^2=\hat P_0$ and $\operatorname{Tr}\hat P_0=r$, we have that
\(
\operatorname{Tr}(\hat A_0^\dagger\hat A_0)=1,
\)
and
\(
F_0^{\text{op}}=1.
\)
The identity $\hat H\hat P_0=E_0\hat P_0$ gives the operator eigenvalue equation, and vectorisation preserves the Hilbert--Schmidt inner product, mapping $\hat H\hat A_0$ to $(\hat I\otimes\hat H)|\hat A_0\rangle\!\rangle$, which proves the final statement. $\square$

Here, the projector remains globally defined and carries the Berry curvature of the Hamiltonian's parametric family.

As an example, we can choose density matrices as one specific case of operator-valued maps. The exact ground-state density operators are
\begin{equation}
 \cD_0(x)=\{\hat\rho\geq0:\Tr\hat\rho=1,\ \hat\rho=\hat P_0\hat\rho\hat P_0\}
 \subset\operatorname{End}(\cH),
 \label{eq:targets}
\end{equation}
which is convex and contains $\hat A_0 \hat A_o^\dagger = \hat P_0 / r$. For the one-qubit family of Eq.~\eqref{eq:hopf}, we use this density-matrix form, and compare its ability to represent the ground state in Figure~\ref{fig:density-lift}. In general any parameterisation that is operator-valued is unobstructed, the operator-valued foundation model in Ref.~\cite{HamiltonZero2026}, where the rank-one obstruction was first identified.


\paragraph{Discussion.---}
We have shown that pointwise expressivity does not yield global expressivity over a gapped, topologically non-trivial Hamiltonian space. The topological obstruction identified here constrains pure-state FNQS independently of any choice of architecture or sampling scheme. Both the pure-state obstruction and Theorem~\ref{thm:lift} also apply in dynamical settings. Replacing $\hat P_0(x)$ by the dynamical projector $\hat Q(x,t)$ gives the corresponding operator-valued target for time evolution, including periodic systems, see Theorems~\ref{thm:nodes} and~\ref{thm:floquet} in the End Matter. There, the obstruction becomes a zero-fidelity world-line for open time intervals, and a node in the Brillouin zone for periodic time intervals. In the Supplementary Material, we also show that a model with limited expressive capacity can produce spurious zero-fidelity nodes in space-time which come in pairs with opposite topological charges.

We emphasize that this no-go has non-zero measure in the parameter space by continuity of the Fidelity function. We see this in the striations of Fig.~\ref{fig:static} and the world-lines of Fig.~\ref{fig:density-lift}(a). The size of the striations depends on how sharply a given model can vary. If $|\Phi(x)$ and the exact ground state are Lipschitz in the FS distance with constants $\Lambda$ and $\kappa$ with respect to a metric on $X$, then the triangle inequality gives $F_0(x)<\tfrac12$ whenever $|x-x_*|<\pi/[4(\Lambda+\kappa)]$. Hence a more expressive or better-trained model can shrink the region by becoming sharper. 

If a model does shrink $\epsilon$ however, the fidelity susceptibility spike will also increase in size. This is because the fidelity-susceptibility is a response function in the neighbourhood of a point. At $x_*$ if the model's training has shrunk the affected region, but the Fidelity is still zero at $x_*$, then the fidelity susceptibility must increase more rapidly. Hence the better the fit a model has on average to the entire Hamiltonian space (including the region around $x_*$), the larger the spurious spikes in fidelity susceptibility will be at the nodes, despite there being no phase transition for these points. This is especially important given a central use for FNQS is to search for new phases of matter via the Fidelity susceptibility method \cite{gu2010fidelity, Rende2025, sinibaldi2026designing}. 

We also emphasize that the location of $x_*$ is model-dependent and not determined by the physics of the Hamiltonian family in question. Furthermore, the obstruction is stable under any deformation of the Hamiltonian's parameter family that preserves the gap and the ground-space rank, since the spectral projectors remain in the same bundle-homotopy class.

In all cases outlined in this work, operator-valued models are unobstructed, since they contain a continuous exact target, which pure-state models do not (Theorem 1). Indeed, the two indices of $\hat A$ must be generated jointly, as in the multilinear sector of Ref.~\cite{HamiltonZero2026} for example. This is why the obstruction is \emph{topological}, independently of expressivity or sampling schemes. Since ground states are pure states by definition, this means operator-valued models must approach approximate purity when they are successful over a Hamiltonian family, but must remain operator-valued should they be able to avoid this obstruction. We note that this does not guarantee an operator-valued model can automatically represent that target class, and the usual expressivity considerations for representability in deep learning apply here \cite{heightman2025deep}.

Finally, we note that $c_r\neq0$ is sufficient but not necessary for every section to vanish, since torsion or higher obstructions can also forbid a nowhere-zero section even when the top Chern class is zero~\cite{MilnorStasheff1974}. Indeed, bundle non-triviality alone is insufficient in rank $r>1$~\cite{MilnorStasheff1974} to identify a general topological obstruction for nodal surfaces in these Hamiltonian families. Future work therefore involves finding the necessary and sufficient criteria, as well as studies on the correlation between expressive capacity and charge-neutral pairs in the nodal surfaces of dynamics (see Supplementary Material) and the development of operator-valued foundation models for ground state problems and time evolution.

\paragraph{Data availability.---}
The code, trained weights, run logs, and dense-grid outputs used to produce all numerical figures are included with this work.

\bibliography{references}

\clearpage
\input{end_matter}
\clearpage
\input{supp_mat}

\end{document}

%% file: end_matter.tex
\section*{End Matter}

\emph{Riesz continuity.---}
Ground state projectors of a continuous Hamiltonian parametric space vary continuously with the Hamiltonian's parameters. Fixing $x_0$ on a finite-dimensional constant-rank gapped stratum, a positively oriented contour $\Gamma$ separates the ground state(s) from the remaining spectrum in a neighbourhood $U$ of $x_0$, and
\begin{equation}
 \hat P_0(x)=\frac{1}{2\pi i}\oint_\Gamma[z\hat I-\hat H(x)]^{-1}\,dz,
 \qquad x\in U.
 \label{em:riesz}
\end{equation}
The resolvent identity and continuity of $\hat H$ imply norm continuity of $\hat P_0$ on $U$. Uniqueness then glues these local projectors.

\emph{Pure state neural realization in Fig.~\ref{fig:static}.---}
The FNQS was a $3$--$48$--$48$--$3$ multilayer perceptron with $\tanh$ activations, receiving the Hamiltonian parameters $(n_x,n_y,n_z)$, and returned three real angles $(a,b,c)$ which we map to a pure-state's $S^3$ coordinates, 
\begin{equation}
    |\psi\rangle = \begin{pmatrix}
        \cos a + i \sin a \cos b \\
        \sin a \sin b e^{ic}
    \end{pmatrix}.
\end{equation}
We minimized the exact mean energy by full-batch Adam for $5000$ steps on a deterministic equal-area grid of $96\times192=18\,432$ points, with a cosine learning-rate decay from $3\times10^{-3}$ to $3\times10^{-5}$. Because this system is small, we can evaluate amplitudes and energies exactly, with the implication that the obstruction is cannot be some sampling artifact. The plotted field was evaluated independently on a $241\times481$ $(\theta,\phi)$ grid of points.

\emph{Operator-valued neural realization in Fig.~\ref{fig:density-lift}.---}
For the operator-valued foundation model, the exact same model as above was used to construct a real vector $\bm v_\theta(\bm n)$ like in Fig.~\ref{fig:static}. This time however, we set 
$\bm r_\theta=\bm v_\theta/\sqrt{|\bm v_\theta|^2}$ 
and $\hat\rho_\theta=(\hat I+\bm r_\theta\cdot\hat{\bm\sigma})/2$, which defines an operator-valued map. The two models therefore each contain exactly $2691$ trainable real parameters, with the only difference being the output representation. After $5000$ steps, an independent $241\times481$ grid gave a minimum fidelity of $99.99985\%$, zero trace and Hermiticity error to machine precision, and a nonnegative spectrum. Continuous refinement from the worst grid points gave a largest refined infidelity $1.5306\times10^{-6}$ and energy error $3.0613\times10^{-6}$.

\begin{figure*}[t]
\centering
\includegraphics[width=\textwidth]{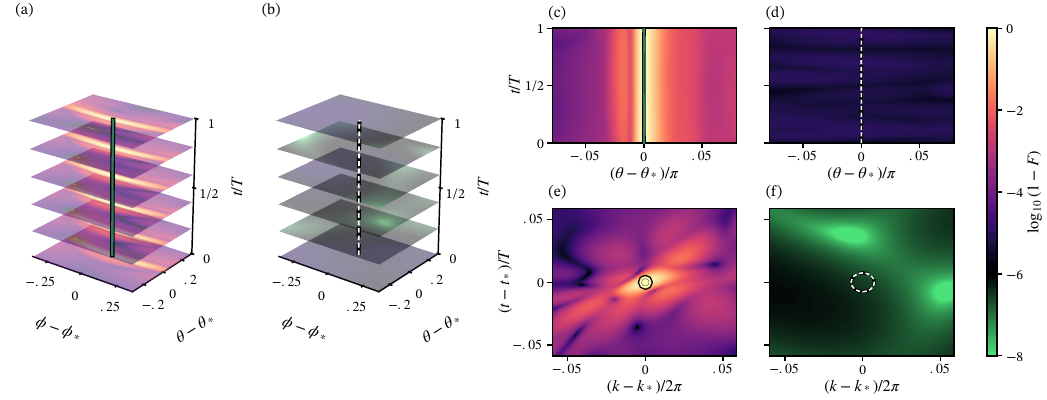}
\caption{
Pure-state FNQS world-lines compared to an operator-valued FNQS, for time evolution of a uniform quench in the main text on a shared heat-map showing $\log_{10}(1-F)$ clipped below $10^{-8}$. Each pair has an identical architecture, optimizer and training schedule, and only the output map changes.(a) Six $(\theta,\phi)$ log-infidelity slices of the finite-time quench for the pure-state model, the green curve is the refined charge $-1$ worldline. (b) A density-matrix-valued FNQS on the same slices, with the absent worldline dashed for reference showing no node forms. The maximum over the full $41\times361\times721$ space-time grid is $1-F=2.000\times10^{-5}$. (c,d) The same comparison profiled over azimuth, $\max_\phi\log_{10}(1-F)$.  (e) The periodic $C=+1$ pure-state FNQS model showing the circled charge $-1$ node, where $1-F=0.9999942$. (f) The density-matrix FNQS model on the same window, with the absent node position dashed. We see in (e) the node's open boundary spanning an area showing it has non-zero measure on the Hamiltonian space.
}
\label{fig:density-lift-dynamics}
\end{figure*}

\paragraph{Dynamical Obstructions.---}
We now consider the time evolution from an initial pure state. Let $\hat Q_0(x)$ be its continuous rank-one projector over the Hamiltonian family $x\in X$, written locally as $\hat Q_0=|\psi_0\rangle\langle\psi_0|$. For $I=[0,T]$, with $\hat U_x(t)=e^{-it\hat H(x)}$, local representatives satisfy $|\psi_t(x)\rangle=\hat U_x(t)|\psi_0(x)\rangle$ with the projector
to the true time-evolved state itself evolving as
\begin{equation}
\hat Q(x,t)=\hat U_x(t)\hat Q_0(x)\hat U_x^\dagger(t).
 \label{eq:true-dynamics}
\end{equation}
Introducing an open time interval introduces no new topology because it is contractible, thus the obstructions above carry over into this dynamical setting, giving the following theorem.

\begin{theorem}
\label{thm:nodes}
Let $x\in X$ with the open time-interval $t\in[0,T]$ define the unitary evolution as
in Eq.~\ref{eq:true-dynamics} such that $L_0=\Ran\hat Q_0\to X$ is topologically non-trivial. 
Then every continuous normalized vector model has some $x_t \in X$ at every time $t$ such that
\[
\hat Q(x_t,t)|\Phi(x_t,t)\rangle=0.
\]
Thus the model's zero-fidelity set intersects every time slice $X\times\{t\}$.
\end{theorem}
\emph{Proof.} For every $t\in[0,T]$, $\hat U_x(t)$ maps $\Ran\hat Q_0(x)$ unitarily and continuously onto $\Ran\hat Q(x,t)$. Hence the time-evolved line bundle is isomorphic to $L_0$, and is therefore topologically non-trivial, at every time. Applying Theorem~\ref{thm:line} separately on each time slice and glueing gives
\(
\hat Q(x_t,t)|\Phi(x_t,t)\rangle=0.
\)
$\square$

Thus at every time $t$, a pure-state FNQS has zero fidelity with the true state at some Hamiltonian parameter $x_t$. 
As an example, we can consider again the one-qubit family $\hat H_0(\bm n)=-\bm n\cdot\hat{\bm\sigma}$ with $\bm n\in S^2$. At the initial-time boundary we have the projector $\hat Q(\bm n,0)=\hat P_0(\bm n)=[\hat I+\bm n\cdot\hat{\bm\sigma}]/2$. Following a uniform quench to $\hat H_{\mathrm d}=\omega\hat\sigma_z/2$, its exact evolution reads 
\begin{equation}
    \hat Q(\bm n,t)=\hat U(t)\hat P_0(\bm n)\hat U^\dagger(t)=[\hat I+(R_z(\omega t)\bm n)\cdot\hat{\bm\sigma}]/2.
\end{equation}  
The target state therefore forms a line bundle over the cylinder $S^2\times[0,T]$. Every time slice is just a rotation of the Hopf bundle and retains $|c_1|=1$, so a transverse zero-fidelity point of an FNQS model becomes a charged worldline crossing every slice, as shown in Figs.~\ref{fig:density-lift-dynamics}(a) and~\ref{fig:density-lift-dynamics}(c).

The finite-time quench leaves the initial and final time boundaries distinct. However, we can also ask whether such an obstruction persists in a periodic setting, where the time interval is closed into a loop and thus the topology changes once more, since loops are not always contractible. Note that this is the setting of Floquet systems~\cite{shirley1965solution, sambe1973steady,eckardt2017colloquium}, where a time-periodic Hamiltonian $\hat H(x,t+T)=\hat H(x,t)$ defines a periodic unitary evolution.

For example, a one-dimensional crystal with momentum $k\in S_k^1$, and the identification $t=0\sim T$ gives the space-time of a torus $S_k^1\times S_t^1$. This space can be topologically non-trivial even though every fixed-time slice is trivial, which motivates the following obstruction in periodically driven systems.

\begin{theorem}
\label{thm:floquet}
Let a periodic true-state line bundle $L\to X\times S_t^1$ have $c_1(L)\neq0$.  No normalized representative of the true state is both global in $x$ and periodic in $t$, and every continuous periodic vector model has a space--time fidelity node, i.e. a point in space-time with zero fidelity to the true trajectory.
\end{theorem}
\emph{Proof.} A representative that is global in $x$ and periodic in $t$ would be a nowhere-vanishing section of $L$, and a continuous periodic model with everywhere-positive fidelity would normalize $\hat Q|\Phi\rangle$ into one; either would trivialize $L$, as in Theorem~\ref{thm:line}. $\square$

This is a separate obstruction, which we can see in the following example. If we cut the torus of the above example at $t=0$, a true-state vector can be chosen on this cylinder, but its endpoints obey
$|\psi(k,T)\rangle=g(k)|\psi(k,0)\rangle$. In a parallel-transported gauge, $g(k)$ is the $U(1)$ loop around the drive cycle, and a periodic gauge exists only when its winding vanishes. On $T^2$, the winding number satisfies
\begin{equation}
 C=\frac{1}{2\pi i}\int_{T^2}\Tr(\hat Q\,d\hat Q\wedge d\hat Q)
   =\frac{1}{2\pi i}\oint_{S_k^1}g^{-1}dg\in\mathbb Z.
 \label{eq:chern}
\end{equation}
This is the familiar topology of a Thouless pump~\cite{Thouless1983}, in which a nonzero winding in $k$ induces a zero-fidelity node in any continuous periodic vector model. Figure~\ref{fig:density-lift-dynamics}(e)-(f) of the End Matter shows this winding across the Brillouin zone in a one-qubit example for a wavefunction- and operator-valued model respectively. We see for both periodic and open time boundaries, continuity of the fidelity function extends the failure over an open region of space--time, and an operator-valued model is again unobstructed by Theorem~\ref{thm:lift}, see Figure~\ref{fig:density-lift-dynamics}.

\emph{Neural realizations in Fig.~\ref{fig:density-lift}.---}
Both dynamical models were $4$--$64$--$64$--$64$--$3$ multilayer perceptrons with $\tanh$ hidden activations and three hyperspherical outputs defining a normalized one-qubit vector in $S^3$, optimized as above. The finite-time network used $(n_x,n_y,n_z,2t/T-1)$ as input, and its $28\times56$ equal-area sphere at $13$ times contained $20\,384$ training points. Figures~\ref{fig:density-lift}(a) and~\ref{fig:density-lift}(c) evaluated on $41\times361\times721$ $(t,\theta,\phi)$ points and a $161\times361$ profiled grid.

The periodic network used $(\cos k,\sin k,\cos t,\sin t)$, and was trained on a $112\times112$ torus for $\bm d(k,t)=(\sin k,\sin t,-1+\cos k+\cos t)$. Fig.~\ref{fig:density-lift}(d-e) uses an independent $641\times641$ evaluation. 

%% file: supp_mat.tex
\clearpage
\setcounter{figure}{0}
\renewcommand{\thefigure}{S\arabic{figure}}
\setcounter{equation}{0}
\renewcommand{\theequation}{S\arabic{equation}}

\onecolumngrid

\section*{Supplemental Material}

In the Supplemental Material we first show that limited expressivity can produce additional oppositely charged node pairs, and then verify numerically that Monte Carlo cannot lift the obstruction. This is because sampling changes the estimator, but the obstruction is at the level of the model class of pure-state vectors. A more expressive model that can vary sharply with the parameter space $X$ can however shrink the size of the neighbourhood around $x_*$, as remarked in the main text.

\emph{Inexpressive models create neutral node pairs.---}

Let $B$ be an oriented smooth $d$-manifold, $E\to B$ a complex rank-$r$ bundle, and $s$ transverse to its zero section.  The preimage theorem makes $Z=s^{-1}(0)$ an oriented submanifold of real codimension $2r$, and its normal bundle is canonically $E|_Z$. The Thom construction~\cite{MilnorStasheff1974} identifies
$[Z]^\vee=e(E_{\mathbb R})=c_r(E).$

For compact $B$ without boundary this is an absolute class. If $B$ has boundary and $s$ is transverse there, $[Z,\partial Z]\in H_{d-2r}(B,\partial B)$ is Poincar\'e--Lefschetz dual to $c_r(E)\in H^{2r}(B)$. A node that reaches this boundary can carry charge out of the domain. In the interior however, a generic creation or annihilation event preserves the total signed charge and therefore produces or removes a neutral pair. In our examples, the total signed node charge on each time slice is $-c_1=-1$. Any transverse nodes beyond the canonical one must
therefore contribute zero net charge, and hence they appear in
oppositely charged pairs. If transversality fails, $Z$ need not be a manifold, although its Euler-class representative persists after a small perturbation \cite{MilnorStasheff1974}.

This also locates the limit of Theorem~\ref{thm:rankr}, where $c_r\neq0$ is sufficient but not necessary for every section to vanish. This is because torsion or higher obstructions can forbid a nowhere-zero section even when the top Chern class is zero, and bundle non-triviality alone is insufficient in rank $r>1$~\cite{MilnorStasheff1974} to identify a general topological obstruction for nodal surfaces in these Hamiltonian families.

By restricting expressivity or optimisation, we can witness these neutral pairs numerically. In Fig.~\ref{fig:supp-pairs}, a width-$8$ model carries up to two extra pairs, with total charge $-1$ on all $241$ time slices of the same dynamics as the one-qubit quench dynamics we have followed in this work. Meanwhile the width-$64$ model of Fig.~\ref{fig:density-lift} carries exactly one node per slice. Hence we see added capacity and training can annihilate the pairs. However, added capacity cannot remove the last node, which vanishes on the operator-valued maps.

\begin{figure}[h]
\centering
\includegraphics[width=\textwidth]{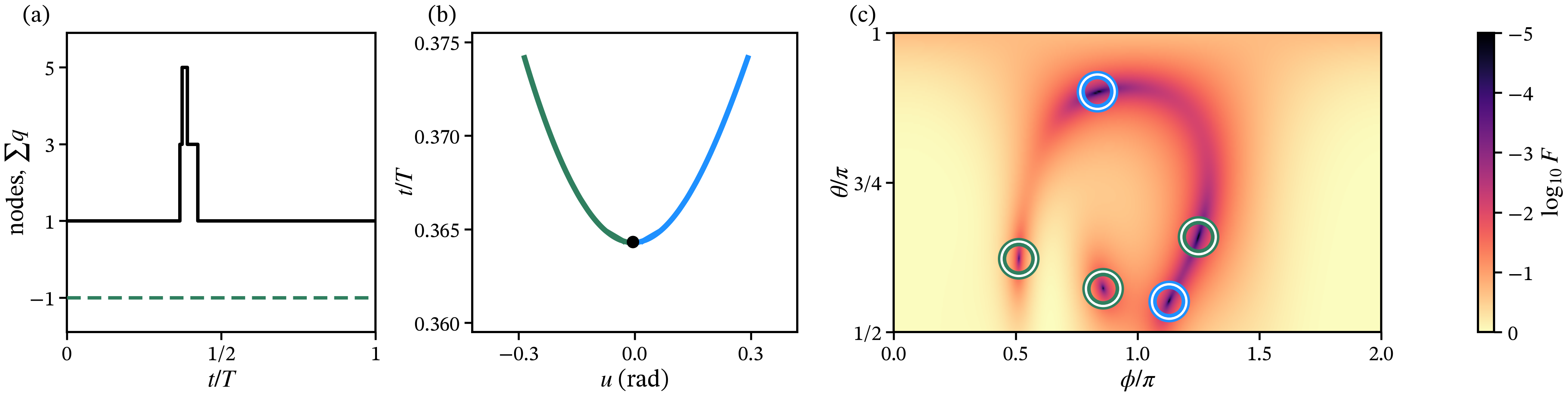}
\caption{\label{fig:supp-pairs} Capacity throttled dynamical models can have pairs of nodes whose opposite topological charge makes them neutral. Here, the finite-time protocol of Fig.~\ref{fig:density-lift}(a,b) is retrained with a single width-$8$ hidden layer ($67$ parameters, $1500$ steps). (a) Refined node count per time slice (solid) and total charge (dashed). The count passes $1\to3\to5\to3\to1$ through neutral creation and annihilation events while $\sum q=-1$ at every one of the $241$ slices. (b) The first birth at $t/T=0.364$, with refined worldlines of the $-1$ (green) and $+1$ (orange) branches in a local coordinate $u$ centered on the event (dot). (c) $\log_{10}F$ on the five-node slice $t/T=0.383$, clipped below $10^{-5}$ so that the node cores are resolved against the $O(1)$-infidelity background; rings mark the refined nodes, colored by charge (green $-1$, blue $+1$), with charges $(-1,-1,-1,+1,+1)$.}
\end{figure}

\emph{Sampling does not lift the obstruction.---}
Theorem~\ref{thm:line} constrains the model class, so any estimator of the model's energy also preserves the zero-fidelity node of Fig.~\ref{fig:static}. This includes when the energy and its gradient are estimated by Monte Carlo sampling instead of being evaluated exactly, as demonstrated numerically in this section on a one-qubit system.

If we sample configurations $s\in\{0,1\}$ by a single-flip Metropolis-Hastings algorithm~\cite{metropolis1953equation} from $p_{\bm n}(s)=|\psi_{\bm n}(s)|^2$ for the wavefunction-valued model, or from the diagonal $p_{\bm n}(s)=\langle s|\hat\rho_{\bm n}|s\rangle$ for the density-matrix model, the training loss for each Hamiltonian parametrised by $\bm n \in S^2$ of the grid becomes a sample mean of the local energy.  We run $16$ independent chains per Hamiltonian from random initial configurations, discard $16$ burn-in sweeps, and keep $64$ sweeps, giving $N=1024$ samples per Hamiltonian per step, with the architecture and training details unchanged from the End Matter.

With $E_{\rm loc}(s)=\sum_{s'}H_{ss'}\psi(s')/\psi(s)$ the wavefunction energy estimate is $\bar E=N^{-1}\sum_i E_{\rm loc}(s_i)$ and its gradient is the standard estimator \cite{heightman2025deep},
\begin{equation}
 \partial_\theta E\simeq\frac{2}{N}\,\mathrm{Re}\sum_{i}\big[E_{\rm loc}(s_i)-\bar E\big]^{*}\,\partial_\theta\log\psi(s_i).
 \label{eq:supp-vmc}
\end{equation}
For the density matrix, $\mathrm{Tr}(\hat\rho\hat H)=\sum_s p(s)E_{\rm loc}(s)$ with $E_{\rm loc}(s)=\sum_{s'}\rho_{ss'}H_{s's}/\rho_{ss}$, and we get,
\begin{equation}
 \partial_\theta E\simeq\frac{1}{N}\sum_{i}\mathrm{Re}\Big\{\big[E_{\rm loc}(s_i)-\bar E\big]\,\partial_\theta\log p(s_i)+\partial_\theta E_{\rm loc}(s_i)\Big\}.
 \label{eq:supp-dm-mc}
\end{equation}
Figure~\ref{fig:supp-sampled} shows the sampled wavefunction model with a zero-fidelity node as Fig.~\ref{fig:static}, and Fig.~\ref{fig:supp-sampled-lift} shows that the sampled density-matrix model reaches the same $10^{-6}$ infidelity as its exactly trained counterpart of Fig.~\ref{fig:density-lift}.

\begin{figure}[h]
\centering
\includegraphics[width=0.90\textwidth]{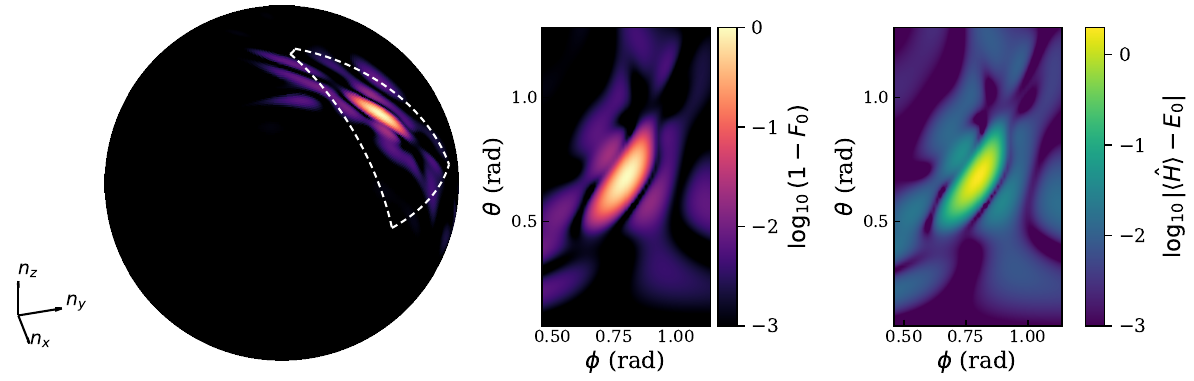}
\caption{\label{fig:supp-sampled}The obstruction of Fig.~\ref{fig:static} under Monte Carlo training, with the same architecture and training as  Fig.~\ref{fig:static}. Here the energy and gradient are estimated from $N=1024$ Metropolis samples per Hamiltonian per step, see Eq.~\eqref{eq:supp-vmc}.  All quantities shown are evaluated exactly on the same independent $241\times481$ grid. The refined maximum sits at $(\theta,\phi)=(0.682,0.799)$ with $F_0=0$ and energy error $2$, matching the exact result of Fig.~\ref{fig:static}.}
\end{figure}

\begin{figure}[h]
\centering
\includegraphics[width=0.78\textwidth]{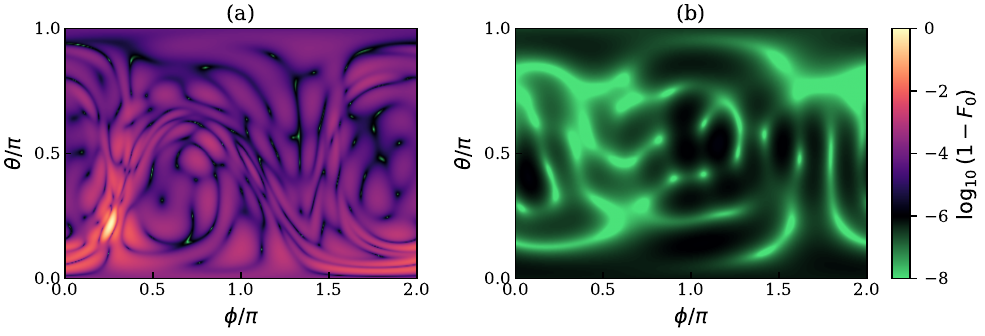}
\caption{\label{fig:supp-sampled-lift} Matched density-matrix lift under Monte Carlo training, as in Fig.~\ref{fig:density-lift}.  (a) The sampled wavefunction-valued model of Fig.~\ref{fig:supp-sampled}, with a maximum $1-F_0=1$. (b) The matched density-matrix model trained with Eq.~\eqref{eq:supp-dm-mc} from the same $N=1024$ samples per Hamiltonian per step, with a maximum $1-F_0=1.44\times10^{-6}$.}
\end{figure}